\documentclass[11pt]{article}
\usepackage[utf8]{inputenc}
\usepackage{fullpage}

\usepackage{lmodern,amsmath,amssymb,amsthm,microtype,tikz,float, booktabs}
\usepackage{dsfont}
\usepackage{thmtools}
\usepackage{algorithm}
\usepackage{algcompatible}
\usepackage{algpseudocode}
\usepackage{float}
\usepackage[hidelinks]{hyperref}
\usepackage{nameref,cleveref}
\usepackage{color}
\usepackage[shortlabels]{enumitem}

\usepackage{csquotes}

\usepackage{svg}

\newtheorem{thm}{Theorem}

\newtheorem{lemma}[thm]{Lemma}
\newtheorem{theorem}[thm]{Theorem}
\newtheorem{corollary}[thm]{Corollary}
\newtheorem{definition}[thm]{Definition}

\crefname{definition}{definition}{definitions}
\Crefname{definition}{Definition}{Definitions}

\newcommand{\define}[1]{\textbf{#1}}

\newcommand{\eps}{\varepsilon}

\newcommand{\R}{\mathbb{R}}
\newcommand{\E}{\mathbb{E}}

\newcommand{\bits}{\{0,1\}}

\newcommand{\cell}[2]{\begin{tabular}[c]{@{}c@{}}#1\\[-3pt]{\footnotesize #2}\end{tabular}}

\algrenewcommand\algorithmicrequire{\textbf{Input:}}
\algrenewcommand\algorithmicensure{\textbf{Output:}}

\DeclareMathOperator{\Lap}{Lap}

\DeclareMathOperator{\mean}{mean}
\DeclareMathOperator{\midr}{mid}
\DeclareMathOperator{\Uniform}{Unif}
\DeclareMathOperator{\MeanSE}{MeanSE}
\DeclareMathOperator{\MaxSE}{MaxSE}
\DeclareMathOperator{\Clip}{Clip}

\DeclareMathOperator{\Loss}{Loss}
\DeclareMathOperator{\ClipLoss}{ClipLoss}

\definecolor{probe}{RGB}{38,98,137}
\definecolor{stepcolor}{RGB}{161,73,28}

\allowdisplaybreaks

\title{Tight Lower Bounds for Differentially Private Continual Counting}

\author{Charlie Harrison \\Google\\ \texttt{\small csharrison@google.com} \and Ethan Leeman \\Google Research\\ \texttt{\small ethanleeman@google.com}}

\begin{document}

\maketitle

\begin{abstract}
The Binary Tree Mechanism is a standard algorithm for differentially private continual counting, but its complete asymptotic optimality under differential privacy has remained unresolved since its introduction. We resolve this question. For $0 \le \eps \le 1$, we prove asymptotically tight lower bounds for $\ell_\infty$, mean squared, and maximum per-coordinate expected squared error simultaneously under all regimes of $0 \le \delta \le 1/4$. These bounds hold for arbitrary mechanisms, even when the entire stream is available in advance. Our lower bounds match the (blocked) Binary Tree Mechanism instantiated with the best of Laplace or Gaussian noise, establishing its asymptotic optimality. Our proof uses a single hard distribution where elements are ``tagged'' as potentially being the target flipped element via a score evaluated down a tree. A simple modification of the score allows the same framework to establish tight lower bounds for all three error measures.
\end{abstract}

\newpage
\section{Introduction}
Differential privacy (DP) \cite{DworkMNS06} provides rigorous privacy guarantees for randomized data analysis. We study the problem of private \emph{continual counting} \cite{dworkcontinual, chan-continual-release}. Given a private bitstream $x \in \{0, 1\}^n$, the mechanism must continually release estimates of all prefix sums $(A_n x)_t = \sum_{i=1}^t x_i$, where $A_n \in \{0, 1\}^{n \times n}$ is the lower-triangular all-ones matrix. We consider this problem under both $\ell_\infty$ and squared error notions. We begin with the former:
\begin{align*}
\alpha_\infty(M) &= \max_{x \in \{0, 1\}^n} \E_M[\| M(x) - A_n x \|_\infty].
\end{align*}
The private continual counting problem has a rich history. It was introduced concurrently by \cite{dworkcontinual} and Chan, Shi, and Song~\cite{chan-continual-release}, who independently proposed the celebrated Binary Tree Mechanism (BTM). When instantiated with Laplace noise it satisfies\footnote{In this section we suppress dependence on $\eps$ for sake of presentation.} $O(\log^2 n)$ error under pure $\eps$-DP, while its Gaussian analogue can achieve $O(\log^{3/2} (n) \sqrt{\log(1/\delta)})$ error under approximate $(\eps, \delta)$-DP.

\cite{dworkcontinual} also established an initial $\Omega(\log n)$ lower bound for $\alpha_\infty$ under pure DP. The approach uses a packing argument on the hard instance where a single contiguous block of $k = \Theta(\log n)$ is flipped on the 0 vector. \cite{andersson2024count} prove this same bound for random rather than worst-case inputs. This lower bound was state of the art until the recent breakthrough of Bairaktari and Larsen \cite{bairaktari2026binary-orig}, who proved an unrestricted $\Omega(\log^{3/2} n)$ lower bound under both privacy regimes. When $\delta = \Theta(1)$ is sufficiently small, their work settles the complexity of continual counting for $\alpha_\infty$. However, for pure DP and for standard regimes of approximate DP where $\delta \le n^{-\Omega(1)}$, a $\sqrt{\log n}$ gap had remained open. Following their initial result and concurrent with our work, Bairaktari, Dahl and Larsen~\cite{bairaktari2026binary}  improved their result to also cover the tight pure DP case. See \Cref{para:concurrent}.

Beyond $\ell_\infty$ error, continual counting is also studied under squared error notions:
\begin{align*}
\mathrm{MeanSE}(M) &= \max_{x \in \{0, 1\}^n} \E_M\left[\frac{1}{n}  \sum_{t=1}^n  (M(x)_t - (A_n x)_t)^2\right],\\
\mathrm{MaxSE}(M) &= \max_{x \in \{0, 1\}^n} \max_{t \in [n]} \E_M\left[(M(x)_t - (A_n x)_t)^2\right].
\end{align*}

MeanSE is of particular importance in machine learning, where continual counting is used as a foundational privacy primitive (e.g.,  \cite{kairouz2021practical}). Note that $\mathrm{MaxSE} \ge \mathrm{MeanSE}$, so any lower bounds on MeanSE automatically transfer. BTM satisfies $O(\log^3 n)$ error for either metric when instantiated with Laplace noise.

Under approximate DP, tight bounds of $\Theta(\log^2 n)$ were given by Henzinger, Upadhyay, and  Upadhyay \cite{henzinger2023almost} under MeanSE for arbitrary mechanisms for fixed $\eps$ and sufficiently small fixed $\delta$. Their upper bound uses a Gaussian instantiation of the matrix-mechanism framework~\cite{matrixmechanism}, and follow-up work improved additive constants in their factorization bounds, without changing their asymptotic order \cite{henzinger2025improved, henzinger2026forc}.

The $\Omega(\log^2 n)$ lower bound also applies to pure DP, leaving a $\log n$ gap relative to the $O(\log^3 n)$ upper bound. Recent work has investigated this gap within the restricted class of Laplace-based matrix mechanisms whose bounds can be reduced to studying properties of matrix factorization costs. Tight results are known  for binary factorizations \cite{arkhipov2026improved} and, up to $\mathrm{polyloglog}$ factors, for MaxSE under arbitrary real factorizations \cite{lin2026near, bulanek2026matrix}. Bhowmik and Hasan \cite{bhowmik2026costs} obtained tight $\Theta(\log^3 n)$ bounds for both MeanSE and MaxSE in this class, resolving the matrix mechanism restricted class but leaving the unrestricted-mechanism question open where a $\log n$ gap persists.

Lastly, \cite{fichtenberger2022constant} provide a lower bound for the worst-case expected \emph{squared} $\ell_\infty$ error $B_\infty(M) = \max_{x \in \{0, 1\}^n} \E_M[\| M(x) - A_n x \|_\infty^2]$ for data-independent mechanisms. Jensen’s inequality yields $B_\infty(M) \ge \alpha_\infty^2(M)$ so our lower bounds on $\alpha_\infty$ immediately yield lower bounds on $B_\infty$ over all mechanisms. We omit a separate discussion of this metric for sake of brevity.

The tightness of our lower bounds is witnessed by a simple generalization of the Binary Tree Mechanism which blocks together $B$ elements in each leaf of the tree and has been explored in the literature in other contexts \cite{chanheavy, kairouz2021practical, zhang-private-stream-processing, pmlr-v267-wang25ea, jain-price-continual, ghazi2022privatecountingdistinctkoccurring}. Blocking adapts the tree depth to the privacy level, improving the error only when privacy is very strong.

\subsection{Our contributions}
Our primary contribution is \Cref{thm:lower_bound}, which gives lower bounds for arbitrary asymptotics of $\eps, \delta$ as functions of $n$:

\begin{theorem}
\label{thm:lower_bound}
Let $n \ge 1, 0 \le \eps \le 1,$ $0 \le \delta \le 1/4$. There exist universal constants $c$ and $C_0$ such that the following holds: Set $\eta = \max\{\eps,\delta\}$ and assume $n \eta \ge C_0$.
Define $L = \log(n \eta)$ and $U = 1 + \min\{L,\log(\eta/\delta)\}$. Then every $(\eps,\delta)$-DP mechanism for binary prefix sums must incur error

\begin{align*}
\alpha_\infty(M) &\geq c \frac{L^{3/2}\sqrt{U}}{\eta},&&&
\mathrm{MaxSE}(M) &\ge \mathrm{MeanSE}(M)  \geq c \frac{L^2 U}{\eta^2}.
\end{align*}
With $\log(\cdot/0) = +\infty$, this gives the pure-DP bound $\Omega(\log^2(n\eps)/\eps)$ and $\Omega(\log^3(n \eps) /  \eps^2)$, respectively.
\end{theorem}

\begin{table}[h]
\centering
\renewcommand{\arraystretch}{1.6}
\begin{tabular}{lccc}
\hline
\textbf{Regime of $\eps$ and $\delta$} & $\alpha_\infty$ & MeanSE/MaxSE & Optimal Mechanism \\
\hline
\cell{$0 \leq \delta \leq 1/n$}{($n\eps \geq C_0$)} & $\Theta\left(\frac{\log^2 (n\eps)}{\eps}\right)$  & $\Theta\left(\frac{\log^3 (n\eps)}{\eps^2}\right)$   & \cell{Blocked Laplacian BTM}{Block size $B=\lceil 1/\eps \rceil$}  \\
\cell{$1/n < \delta \leq \eps/2$}{($n\eps \geq C_0$)} & $\Theta\left(\frac{\log^{3/2} (n\eps) \sqrt{\log(\eps/\delta)}}{\eps}\right)$  & $\Theta\left(\frac{\log^{2} (n\eps) \log(\eps/\delta)}{\eps^2}\right)$   & \cell{Blocked Gaussian BTM}{Block size $B=\lceil 1/\eps \rceil$} \\
\cell{$\eps/2 < \delta \leq 1/4$}{($n \cdot \max \{\eps,\delta\} \geq C_0$)} & $\Theta\left(\frac{\log^{3/2} (n\delta)}{\delta}\right)$  & $\Theta\left(\frac{\log^2 (n\delta)}{\delta^2}\right)$   & \cell{Blocked Gaussian BTM}{Block size $B=\lceil 1/\delta \rceil$}  \\
$n \cdot \max\{\eps,\delta\} < C_0$ & $\Theta(n)$  & $\Theta(n^2)$   & All Zeroes Mechanism \\
\hline
\end{tabular}
\caption{Tight bounds of $\alpha_\infty$, MaxSE, and MeanSE across all regimes of $0 \leq \delta \leq 1/4, 0 \leq \eps\leq 1$ for \emph{arbitrary mechanisms} under sufficiently large $n$. \Cref{thm:lower_bound} gives the first three rows' lower bounds, each of which is achievable with the appropriate Blocked Binary Tree Mechanism. The proofs that the appropriate Blocked Binary Tree Mechanisms achieve the optimal error rates, as well as the lower bound for the high privacy regime $n \cdot \max\{\eps, \delta\} < C_0$, are left to Appendix~\ref{app:matching-upper-bounds}. If $\max\{\eps,\delta\} > \frac{1}{n^{1-c}}$ for some constant $c>0,$ then one can use the standard Binary Tree Mechanism to achieve the optimal rates and ignore the diversion to the Blocked BTM. See Appendix~\ref{app:matching-upper-bounds}, particularly \Cref{lem:vanilla-ok-c}.} 
\label{tab:bounds}
\end{table}

Particularly interesting cases are pure DP and approximate DP with $\delta = n^{-\Omega(1)}$ for fixed $\eps$, where
\Cref{thm:lower_bound} closes the $\sqrt{\log n}$ gap for continual counting under $\alpha_\infty$, as well as the $\log n$ gap for MeanSE (and hence MaxSE). Both lower bounds hold for arbitrary mechanisms.

While previous work resolved the fixed $\delta = \Theta(1)$ regime, arguably the $\delta = n^{-\Omega(1)}$ regime is more natural. Indeed, in the DP literature, it is commonly recommended to set $\delta \ll 1/n$ \cite{dworkboosting, dworkrothbook,
complexitydp, 
near2025guidelines}, to limit the risk of catastrophic privacy failure.\footnote{The aptly named ``catastrophe mechanism'' releases the raw data of a single person at random, yet it trivially achieves $(0, 1/n)$-DP.}
In this regime, our lower bounds establish that BTM instantiated with Laplace noise is asymptotically optimal under all three error notions, even under approximate DP. This aligns with the observations of \cite{andersson2024count}, whose improved pure-DP mechanism achieves a smaller mean squared error bound than the Gaussian matrix mechanism of \cite{henzinger2023almost} for sufficiently small $\delta$.
Our bounds hold even in the easier offline setting \cite{cohenlowerbounds-continual} where the entire stream is known in advance, and therefore also apply to the online setting.

\begin{table}[ht]
\centering
\renewcommand{\arraystretch}{1}
\begin{tabular}{lcccc}
\hline
\textbf{Regime of $\delta$} & $\alpha_\infty$ & Ref & MeanSE/MaxSE & Ref \\
\hline
$\delta = 0$ & $\Theta(\log^2 (n))$  &  \cite{bairaktari2026binary}${\dagger}$  & $\Theta(\log^3 (n))$ &  \\
$\delta \leq n^{-\Omega(1)}$ & $\Theta(\log^2 (n))$ &   & $\Theta(\log^3 (n))$ &  \\
$n^{-o(1)} \le \delta = o(1)$ & 
	$\Theta(\log^{3/2}(n) \sqrt{\log (1/\delta)})$ & &
	$\Theta(\log^2 (n) \log (1/\delta))$& \\
$\delta = \Theta(1)$ & $\Theta(\log^{3/2} (n))$ & \cite{bairaktari2026binary-orig} & $\Theta(\log^2 (n))$   & \cite{henzinger2023almost} \\
\hline
\end{tabular}
\caption{Special cases of \Cref{tab:bounds} in the literature for fixed $\eps$. The $\dagger$ indicates that \cite{bairaktari2026binary} was produced concurrently with our results; see \Cref{para:concurrent}. Previously only the $\delta = \Theta(1)$ case was completely characterized, and we complete the understanding for $0 \le \delta \le o(1)$. The $\delta = \Theta(1)$ regime applies to a small enough fixed $\delta$ in order to satisfy assumptions in \cite[Theorem 1]{bairaktari2026binary} or \cite[Theorem 4]{henzinger2023almost}.} 
\label{tab:specbounds}
\end{table}

\paragraph{Concurrent work.}\label{para:concurrent} Bairaktari and Larsen \cite{bairaktari2026binary-orig} initially published tight results for this problem under approximate differential privacy for fixed $\delta$ and for $\alpha_\infty$. Their results implied an improved (but loose) bound of $\Omega(\log^{3/2} n)$ under pure DP. Concurrent with our work, Bairaktari, Dahl and Larsen \cite{bairaktari2026binary} released an improved version of their results which achieves the tight bounds for pure DP as well. Our results all share a similar tree structure. Where we create a 4-ary tree $\mathcal{T}$ and try to observe a shift across children, they create a binary tree and try to observe a shift across grandchildren. In both cases, 4 nodes are considered to allow for probing opposite-parity subtrees which exclude the changed input location. Their approach first reduces the problem to input-independent additive noise in a tree basis, and analyzes conditional residual noise along the tree. 
In our work, we avoid their reduction to input-independent noise. We achieve this by observing aggregate statistics defined by internal nodes of a 4-ary tree whose leaves represent the (possibly shifted) error of the mechanism. These statistics, which we call probe functions,  give rise to an edge weighting of the 4-ary tree. Martingales defined by these weighted 4-ary trees can be directly compared to both the expected error as well as the privacy parameters via a telescoping potential and group privacy.

\section{Preliminaries}

We will always use $\log$ for the natural logarithm $\log_{e}$. We write $[d] := \{1, \dots, d\}$ and $\mathbf{1} \in \mathbb{R}^d$ as the all-ones vector. For a vector $v\in\mathbb R^d$ and $t\in[d]$, write $v_{[t]}:=(v_1,\ldots,v_t)\in\mathbb R^t$ for the restriction of $v$ to its first $t$ coordinates. We use the convention $\log(\cdot /0)=+\infty$.

\paragraph{Privacy.} We begin by formalizing the privacy notion specific to our setting and stating a useful group privacy lemma.

\begin{definition}[\cite{dworkrothbook}]
A randomized mechanism $M \colon \{0, 1\}^n \to \mathcal{Y}$ satisfies $(\eps, \delta)$-differential privacy (DP) if $\Pr(M(x) \in S) \le e^{\eps} \Pr(M(x') \in S) + \delta$ for all $x, x' \in \{0, 1\}^n$ with Hamming distance $d_H(x, x') \le 1$ and all measurable $S \subseteq \mathcal{Y}$. When $\delta = 0$, $M$ satisfies pure $\eps$-DP.
\end{definition}

\begin{lemma}[Group privacy \cite{complexitydp}]
\label{lem:group-privacy}
If $M: \bits^n \to \mathcal{Y}$ is $(\varepsilon, \delta)$-DP, then for any $x, x' \in \{0, 1\}^n$ with $d_H(x, x') \le k$, and all $S \subseteq \mathcal{Y}$,
\[
\Pr(M(x) \in S) \le e^{k \eps} \Pr(M(x') \in S)
+ \delta_k,
\]

where we use the shorthand
\[\delta_k = \delta \sum_{i=0}^{k-1} e^{i \eps} = \begin{cases}
\delta \frac{e^{k\varepsilon} - 1}{e^\varepsilon - 1} & \eps > 0\\
\delta k & \eps = 0.
\end{cases}
\]
\end{lemma}

\paragraph{Probe and potential functions.}
For $v \in \mathbb{R}^d$ and non-empty $u \subseteq [d]$,
a \emph{probe} assigns a representative value to the coordinates
$(v_i)_{i \in u}$, while its associated \emph{potential} informally measures
their `spread.' We use two probe--potential pairs: \{mean, variance\} and \{midrange, range\}.

\begin{table}[ht]
\centering
\renewcommand{\arraystretch}{1.6}
\begin{tabular}{lll}
\hline
\textbf{Probe} & \textbf{Potential} & Use-case \\ \hline
$\mean = \mu_u(v) = \frac{1}{|u|} \sum_{i \in u} v_i$ & $V_u(v) = \frac{1}{|u|} \sum_{i \in u} (v_i - \mu_u(v))^2$ & $\MeanSE$ bounds \\[6pt]
$\midr = c_u(v) = \frac{1}{2} \left(\max_{i \in u} v_i + \min_{i \in u} v_i\right)$ & $r_u(v) = \max_{i \in u} v_i - \min_{i \in u} v_i$ & $\alpha_\infty$ bounds \\ \hline
\end{tabular}
\label{tab:probes_potentials}
\end{table}

Both potentials are nonnegative. Both probes are translation
equivariant, while their potentials are translation invariant:
for either pair $(\phi,P) \in \{(\mu,V),(c,r)\}$ and any
$\gamma \in \mathbb{R}$,
\begin{equation}\label{eq:equivariance}
    \phi_u(v + \gamma\mathbf{1}) = \phi_u(v) + \gamma,
    \qquad
    P_u(v + \gamma\mathbf{1}) = P_u(v).
\end{equation}

We will use the following lemmas relating probes with their associated potentials.
\begin{lemma}[Midrange stability]
\label{lem:midrange-stability}
For any $v \in \mathbb{R}^d$ and non-empty $u' \subseteq u \subseteq [d]$,
\[
\left| c_{u'}(v) - c_u(v) \right| \le \frac{r_u(v) - r_{u'}(v)}{2}.
\]
\end{lemma}
\begin{proof}
From $u' \subseteq u,$ we have $\min_{i \in u} v_i \leq \min_{i \in u'} v_i$ and $\max_{i \in u'} v_i \leq \max_{i \in u} v_i$. Now suppressing the argument $v$, this can be written as:
\[
[c_{u'}-r_{u'}/2,\;c_{u'}+r_{u'}/2]
\subseteq
[c_u-r_u/2,\;c_u+r_u/2].
\]
Comparing both endpoints gives
$|c_{u'}-c_u|\le(r_u-r_{u'})/2$.
\end{proof}

\begin{lemma}[Law of total variance]
\label{lem:variance-decomposition}
Let $v \in \mathbb{R}^d$, and let $u_1,\ldots,u_B$
partition a non-empty set $u \subseteq [d]$ into non-empty
parts. Writing $w_i = |u_i|/|u|$, we have
\[
    V_u(v) - \sum_{i=1}^B w_i V_{u_i}(v)
    =
    \sum_{i=1}^B w_i
    \bigl(\mu_{u_i}(v)-\mu_u(v)\bigr)^2.
\]
\end{lemma}

\begin{proof}
For $j\in u_i$, expand
$v_j-\mu_u(v)=(v_j-\mu_{u_i}(v))+(\mu_{u_i}(v)-\mu_u(v))$.
Summing squares over each part eliminates the cross terms,
since $\sum_{j\in u_i}(v_j-\mu_{u_i}(v))=0$.
Dividing by $|u|$ gives the identity.
\end{proof}

We will use the fact that (a function of) a flipped uniform bit string and its flip index are independent.
\begin{lemma}\label{lem:bit-flip-iid}
    Let $X \sim \Uniform(\bits^n)$ and $J \sim \Uniform([n])$ be independent, and let $f$ be a randomized function whose internal randomness is independent of $(X, J)$. Then $J$ and $f(X \oplus e_J)$ are independent.
\end{lemma}
\begin{proof}
    Since $X \sim \Uniform(\bits^n)$ is independent of $J$ and $x \mapsto x \oplus e_j$ is a bijection, $X' = X \oplus e_J$ is distributed as $\Uniform(\bits^n)$ conditioned on $J = j$ for any $j \in [n]$. Furthermore, because the randomness of $f$ is independent of $(X,J)$, $J$ is independent of $f(X')$.
\end{proof}

\paragraph{Martingales.} The proof of our theorem requires some standard results from martingale theory. A reader not familiar with martingales can understand the proofs if they mentally replace `martingale' with `1-dimensional, mean 0 random walk' and accept the lemmas given. We give the core definition and results we need now.

\begin{definition}[Discrete Finite-Time Martingale \cite{Williams1991}] 
\label{def:martingale}
A \define{martingale}\footnote{We leave out the measure-theoretic components of defining a martingale as a stochastic process adapted to a filtration of a probability space. For a richer treatment, see \cite{Williams1991} or \cite{koralov2007theory} Section 13.} is a sequence of random variables $X_0,X_1, X_2, \ldots, X_T$ where $\E[|X_t|] < \infty$ and for all $0 \leq t < T$: \[\E[X_{t+1} | X_1, \ldots, X_t] = X_t.\]

\end{definition}

We state two well known facts about the behavior of martingales:
\begin{lemma}[Azuma's Inequality with identical bounds \cite{azuma1967weighted}, \cite{Williams1991} E14.2]\label{lem:azuma}
Suppose $X_t$ is a martingale with $X_0=0$. Let $\Delta_i = X_i - X_{i-1}$ so that $X_t = \sum_{i=1}^t \Delta_i.$ Suppose that $|\Delta_i| \leq a$ almost surely. Then \[\Pr(X_T \geq c) \leq e^{\frac{-c^2}{2Ta^2}}.\]
\end{lemma}

\begin{lemma}[\cite{Williams1991} Section 12.1] \label{lem:martingale-variance}
Suppose $X_0, X_1, \ldots, X_T$ is a martingale where each $X_t$ has finite variance and $X_0=0$. Let $\Delta_i = X_i - X_{i-1}$ so that $X_t = \sum_{i=1}^t \Delta_i.$ Then \[\E[X_t^2] = \sum_{i=1}^t \E[\Delta_i^2].\]
\end{lemma}
\section{Our lower bounds}
\subsection{Proof overview}
Our identification argument is related to the randomized tree decoder of ~\cite{ghazi-on-trees}, which combines evidence along a root-to-leaf path with group privacy.

To construct our instance, sample a uniform string $X \sim \Uniform(\{0,1\}^m)$ and apply $E_k$ which repeats each bit $k$ times. Comparison inputs are formed by choosing an independent uniform leaf $J$ and flipping its bit, giving $X' = X \oplus e_J$ and $E_k(X')$. This mirrors the approach by \cite{dworkcontinual} but replaces their 0-vector with a uniformly random vector.\footnote{The technique of using random inputs to make a probabilistic argument to construct a lower bound was first used on this problem by \cite{andersson2024count}.}

We build a 4-ary tree\footnote{In the proof overview we assume for simplicity that $n = k 4^H$ for some tree height $H$ to avoid the need for padding.} around the \emph{shifted normalized residual} $Y = \frac{M(E_k(X)) - A_n E_k(X')}{k}$, which compares the mechanism output on the \emph{original stream} $X$ against the exact prefix sums of the \emph{flipped stream} $X'$. Because the mechanism answered queries for $X$ but $Y$ subtracts $X'$, before block $J$ the residual is simply the mechanism's normalized error. After block $J$, the residual is the mechanism's normalized error \emph{plus} a step of size $\pm 1$, and within block $J$ the added shift interpolates 0 and $\pm 1$.

The proof involves scoring each leaf $j$ of the tree with a function $S(j, Y)$, and tagging all the leaves whose score exceeds a threshold.\footnote{As intuition, consider a player who is given $M(E_k(X))$ and $X'$ and whose goal is to guess where the shift is.} The result follows from bounding the probability that the target leaf $J$ is tagged; privacy gives an upper bound and accuracy gives a lower bound.

To define our scoring function, we construct edge weightings on the 4-ary tree from $Y$ where larger weights attempt to target $J$. These weights are defined via cross-parity \emph{probe functions} that attempt to detect if $J$ is present between the two probed children or not. Precisely, the probes measure the step signal between the children and give a clipped value to each edge. 
The step signal will either be $+1$ or $-1$ (depending on the bit being flipped) if $J$ lies between the two probed children and $0$ otherwise, so we need to handle those two cases separately. We do so by considering the \emph{max} of two candidate scores on two separate edge weightings, one targeting the $+1$ case and one targeting the $-1$ case. To measure these two different types of step differences, we need to flip the edge weights depending on the sign of the bit and center to $1/2$ to find weights that are centered.
The edge weighting has the property that for every parent node, the sum of edge weights to its children is 0. From this edge weighting, we define the score as simply the sum of the edges along the path from the root to the leaf $J$. Note that the cross-parity probes save us from handling the complexity of probing $Y$ on the path-to-$J$ which involves the $0 \to \pm1$ interpolation. This also explains why we need a 4-ary tree: we need more than 2 children to get a ``clean'' probe difference, and 4 children keeps everything symmetric.

The lower bound proceeds by sandwiching $\Pr[S(J, Y) \geq \text{Threshold}]$ from both sides:
\begin{itemize}
    \item \textbf{Accuracy lower bound (\Cref{lem:accuracy-bound}):} Because $Y$ contains the step signal, the weighting ensures $\Pr[S(J, Y) \geq \text{Threshold}]$ is high unless excessive clipping occurs, which we compare to the $\alpha_\infty$ or $\MeanSE$ loss via a telescoping potential. This directly shows that low loss implies tagging $J$ with high probability.
    \item \textbf {Privacy upper bound (\Cref{lem:privacy-bound}):} By group privacy, replacing the mechanism input $E_k(X)$ with the flipped input $E_k(X')$ gives us an upper bound on tagging $J$, and supplies us with a tree with zero step signal in its error residual $Y'$. Since $X'$ is independent of $J$, $Y'$ carries no information about $J$, and the result follows by bounding $\Pr[S(J, Y') \geq \text{Threshold}]$ via a martingale analysis.
\end{itemize}
Balancing the accuracy lower bound against the privacy upper bound forces the mechanism error to be large, establishing the tight lower bounds. The only difference between our $\alpha_\infty$ and $\MeanSE$ bounds is in the probe functions used.

\begin{figure}[h]
  \centering
  \resizebox{\linewidth}{!}{\definecolor{tdrust}{RGB}{184, 78, 26}
\definecolor{tdblue}{RGB}{26, 95, 180}
\definecolor{tdink}{RGB}{46, 52, 54}

\begin{tikzpicture}[
    x=1cm, y=1cm, yscale=-1,
  td edge/.style={draw=tdink!55, line width=0.55pt},
  td path/.style={draw=tdrust, line width=1.45pt},
  td name/.style={circle, draw=tdink!85, fill=white, inner sep=0pt, minimum size=15pt, font=\fontsize{8}{9}\selectfont},
  td dot/.style={circle, draw=tdink!80, fill=white, inner sep=0pt, minimum size=3pt},
  td tiny/.style={font=\fontsize{8}{8}\selectfont},
  td small/.style={font=\fontsize{8.5}{9.5}\selectfont},
  td label/.style={font=\fontsize{9}{10}\selectfont\bfseries, text=tdrust}
  ]

  \foreach \idx in {0,...,63} {
    \pgfmathsetmacro{\xx}{1.20 + \idx * 0.190}
    \coordinate (td-3-\idx) at (\xx, 3.35);
  }

  \foreach \idx in {0,...,15} {
    \pgfmathsetmacro{\xx}{1.485 + \idx * 0.760}
    \coordinate (td-2-\idx) at (\xx, 2.50);
  }

  \foreach \idx in {0,...,3} {
    \pgfmathsetmacro{\xx}{2.625 + \idx * 3.040}
    \coordinate (td-1-\idx) at (\xx, 1.65);
  }

  \coordinate (td-0-0) at (7.185, 0.75);

  \foreach \child in {0,...,3} {
    \draw [td edge] (td-0-0) -- (td-1-\child);
  }

  \foreach \parent in {0,...,3} {
    \foreach \child in {0,...,3} {
      \pgfmathtruncatemacro{\target}{int(4*\parent+\child)}
      \draw [td edge] (td-1-\parent) -- (td-2-\target);
    }
  }

  \foreach \parent in {0,...,15} {
    \foreach \child in {0,...,3} {
      \pgfmathtruncatemacro{\target}{int(4*\parent+\child)}
      \draw [td edge] (td-2-\parent) -- (td-3-\target);
    }
  }

  \draw [td path] (td-0-0) -- (td-1-2);
  \draw [td path] (td-1-2) -- (td-2-9);
  \draw [td path] (td-2-9) -- (td-3-39);

  \node [td tiny, anchor=south east] at (4.70, 1.15) {$w^\tau(u,u_0)$};
  \node [td tiny, anchor=center, fill=white, inner sep=1pt] at (6.36, 1.22) {$w^\tau(u,u_1)$};
  \node [td tiny, text=tdrust, anchor=center, fill=white, inner sep=1pt] at (8.01, 1.22) {$w^\tau(u,u_2)$};
  \node [td tiny, anchor=south west] at (9.67, 1.15) {$w^\tau(u,u_3)$};

  \node [td name, draw=tdrust, line width=0.7pt] at (td-0-0) {$[m]$};

  \node [td name] at (td-1-0) {$u_0$};
  \node [td name, draw=tdblue, fill=tdblue!10, line width=0.7pt] at (td-1-1) {$u_1$};
  \node [td name, draw=tdrust, text=tdrust, line width=0.8pt] at (td-1-2) {$u_2$};
  \node [td name, draw=tdblue, fill=tdblue!10, line width=0.7pt] at (td-1-3) {$u_3$};

  \node [anchor=east, td tiny, text=tdblue] at (5.32, 1.65) {probe $\phi_{u_1}$: $+0$};
  \node [anchor=west, td tiny, text=tdblue] at (12.05, 1.65) {probe $\phi_{u_3}$: $+s$};

  \foreach \idx in {0,...,15} \node [td dot] at (td-2-\idx) {};
  \node [td dot, draw=tdrust, fill=tdrust, minimum size=4.5pt] at (td-2-9) {};

  \foreach \idx in {0,...,63} \fill [tdink!80] (td-3-\idx) circle [radius=0.65pt];
  \fill [tdrust] (td-3-39) circle [radius=1.8pt];

  \node [td tiny, anchor=east, text=black!70] at (0.60, 0.35) {depth};
  \node [td small, anchor=east] at (0.60, 0.75) {$0$};
  \node [td small, anchor=east] at (0.60, 1.65) {$1$};
  \node [td small, anchor=east] at (0.60, 2.50) {$2$};
  \node [td small, anchor=east] at (0.60, 3.35) {$H{=}3$};

  \node [anchor=north, td label] at (8.61, 3.38) {$j$};
  \node [anchor=north, td small] at (7.185, 3.70) 
    {$S(j,y) = \max_{\tau \in \{\pm 1\}} \sum_{(u, u_i) \in \mathrm{path}(j)} w^{\tau}(u, u_i)$, \quad where $w^{\tau}([m],u_2) = \mathrm{Clip}_{[-a,a]}\!\big(\tau \cdot d_u^{13}(y) - \tfrac{1}{2}\big)$ and $d_u^{13}(y) = \phi_{u_3}(y) - \phi_{u_1}(y)$.};

\end{tikzpicture}}
  \vspace{-20pt}
\caption{Whole-tree decoder ($H = 3$ shown). Each internal node has edge weights that sum to 0. Our privacy and accuracy bounds are obtained from analyzing the score of $J$ against a residual which includes the signed shift.} \label{fig:whole-tree-decoder}
\end{figure}

\subsection{Tree setup and definitions}\label{sec:tree}
Let $m, k, n$ be positive integers such that $q := km \le n$, where $m = 4^H$ for an integer $H \ge 1$. We define the following components:
\begin{itemize}
\item \textbf{Expansion operator:} $E_k \colon \{0,1\}^m \to \{0,1\}^n$ repeats each bit $k$ times and applies padding\footnote{This padding minimally affects our analysis. Our tree lemmas operate only on the elements covered by the tree, and the proof of \Cref{thm:lower_bound} only needs $q \ge n/4$ for the $\MeanSE$ analysis.} when $q < n$:
\[
(E_k(x))_t = \begin{cases} x_{\lceil t/k \rceil} & \text{if } 1 \le t \le q, \\ 0 & \text{if } q < t \le n. \end{cases}
\]

\item \textbf{The $4$-ary tree $\mathcal{T}$:} The complete tree of height $H$. The root represents the full block index set $[m]$. Each internal node $u$ corresponds to a block range $[\ell, r] \subseteq [m]$ of length $L = r- \ell + 1$ (a power of $4$), with four children partitioning $[\ell, r]$ into quarters:
\begin{align*}
u_0 &= [\ell, \, \ell + L/4 - 1], & u_1 &= [\ell + L/4, \, \ell + 2L/4 - 1], \\
u_2 &= [\ell + 2L/4, \, \ell + 3L/4 - 1], & u_3 &= [\ell + 3L/4, \, r].
\end{align*}
Each leaf $j \in [m]$ corresponds to the $j$-th coordinate block $I_j = \{(j-1)k + 1, \dots, jk\}$. More generally, each node $u = [\ell, r]$ is associated with the stream coordinate interval $I_u = \bigcup_{j=\ell}^r I_j = \{(\ell-1)k+1, \dots, rk\} \subseteq [q]$.

For any vector $y \in \R^q$, we write probe functions $\phi$ and potential functions $P$ as 
$\phi_u(y) :=\phi_{I_u}(y)$ and $P_u(y) := P_{I_u}(y)$.

\item \textbf{Edge Weights:} Let $a>1$ and define 
\[\Clip_{[-a,a]}(x) = \begin{cases} -a & \text{if } x < -a, \\ x & \text{if } -a \leq x \leq a, \\ a &\text{if } a < x.\end{cases}\]

For both picks of $\phi \in \{\midr, \mean \}$, we will define two edge weightings $w^{\phi,\tau}$ for the tree $\mathcal{T}$ indexed by $\tau \in \{1,-1\}.$ At each internal node $u \in \mathcal{T}_{\mathrm{int}}$, we assign weights $w^{\phi,\tau}(u,u_i)$ for the 4 edges from $u$ to its children $u_i$ based on $y \in \mathbb{R}^q$:

Let \[
d_{u, \phi}^{02}(y) = \phi_{u_2}(y) - \phi_{u_0}(y), \qquad d_{u, \phi}^{13}(y) = \phi_{u_3}(y) - \phi_{u_1}(y).
\]
We drop $\phi$ from the notation unless we are specifically instantiating $\phi \in \{\midr, \mean\}$.

For $\tau \in \{1, -1\}$, assign weights
\begin{align*}
w^{\tau}(u,u_0) = -&\Clip_{[-a,a]}(\tau \cdot d_{u}^{13}(y) - \frac{1}{2}), & w^{\tau}(u,u_1) = \quad &\Clip_{[-a,a]}(\tau \cdot d_{u}^{02}(y) -\frac{1}{2}), \\
w^{\tau}(u,u_2) = \quad&\Clip_{[-a,a]}(\tau \cdot d_{u}^{13}(y) - \frac{1}{2}), & w^{\tau}(u,u_3) = -&\Clip_{[-a,a]}(\tau \cdot d_{u}^{02}(y) -\frac{1}{2}).
\end{align*}

\item \textbf{Weight sum:} Let $\mathrm{path}(j)$ denote the set of $H$ parent-child pairs on the unique path from the root to leaf $j$ and let $\mathrm{path}_h(j) \subseteq \mathrm{path}(j)$ be the subset of pairs for which the child is at depth $\leq h$. The weight at depth $h$ of this path for index $\tau$ is 
\[
S_{\tau}^h(j, y) = \sum_{(u,u_i) \in \mathrm{path}_h(j)} w^{\tau}(u,u_i)
\]
with final weight
\[
S_{\tau}(j, y) = \sum_{(u,u_i) \in \mathrm{path}(j)} w^{\tau}(u,u_i).
\]
Our object of interest will be the max weight over $\tau \in \{-1, 1\}$
\[
S(j,y) = \max_{\tau \in \{\pm 1\}} S_{\tau}(j,y).
\]
\end{itemize}

Note that $S$ and $w$ are also parameterized by $\phi$, which we also drop unless specifically instantiating a probe function.

\begin{lemma}\label{lem:tree-is-martingale}
Let $w$ be an edge weighting of $\mathcal{T}$. Suppose for any $v \in \mathcal{T}_\mathrm{int}$ that the sum over edges that are children of $v$ is $\sum_{(v,u') \in \mathcal{T}} w(v,u') = 0.$ Let $J \sim \Uniform([4^H])$ select a random leaf and consider the random variable $Z_t = \sum_{(u,u') \in \mathrm{path}_t(J)} w(u,u')$ for $t=0,1,\ldots,H.$ Then $Z_t$ is a martingale.
\end{lemma}
\begin{proof}
    Because $\sum_{(v,u') \in \mathcal{T}} w(v,u') = 0$, for any fixed internal node $v$, selecting a random child adds expected weight $0$ to the sum. Since $J$ is selected uniformly at random, conditional on the path ctaken so far, the next child is uniform among the children of the
current node. Hence  $\E[Z_t | Z_0, \ldots, Z_{t-1}] = Z_{t-1}$ for $t= 1,\ldots, H.$
\end{proof}

\subsection{Tree lemmas}
\begin{lemma}[Accuracy lower bound] \label{lem:accuracy-bound}
Let $n, H, k$ be positive integers with $q = k 4^H \le n$, and let $m = 4^H$. Fix $X \in \{0,1\}^m$, $W \in \R^n$ and sample $J \sim \Uniform([m])$, setting $X' = X \oplus e_J$. Define the shifted residual $Y = \frac{(W - A_n E_k(X'))_{[q]}}{k}$. Then for any $a > 1:$
\begin{align*}
\Pr_{J}[ S_{\midr}(J,Y) \geq H/4] & \ge 1 - \frac{4a}{\sqrt{H}} - \frac{4R}{kH(a-1)}, \\
\Pr_{J}[ S_{\mean}(J,Y) \geq H/4] & \ge 1 - \frac{4a}{\sqrt{H}} - \frac{8Q}{k^2H(a-1)^2},
\end{align*}
where $R = \|W - A_n E_k(X)\|_\infty$ and $Q = \frac{1}{q}\sum_{t \in [q]}(W - A_n E_k(X))_t^2$.
\end{lemma}

\begin{proof}
For a fixed $X$ and $W$ and $J\sim \Uniform([m])$, the shifted residual $Y$ depends on $J$. However, the \emph{non-shifted} error $Z=\frac{(W-A_nE_k(X))_{[q]}}{k}$ is independent of $J$. 

We construct the \emph{uncentered} tree weighting:
\begin{align*}
\tilde{w}(u,u_0) = -&\Clip_{[-a,a]}(d_{u}^{13}(z)), & \tilde{w}(u,u_1) = \quad &\Clip_{[-a,a]}(d_{u}^{02}(z) ), \\
\tilde{w}(u,u_2) = \quad&\Clip_{[-a,a]}(d_{u}^{13}(z)), & \tilde{w}(u,u_3) = -&\Clip_{[-a,a]}(d_{u}^{02}(z)),
\end{align*}
which gives rise to mirroring definitions:
\begin{eqnarray*}
\tilde{S}^h(j, z) &=& \sum_{(u,u_i) \in \mathrm{path}_h(j)} \tilde{w}(u,u_i) \\
\tilde{S}(j, z) &=& \sum_{(u,u_i) \in \mathrm{path}(j)} \tilde{w}(u,u_i).
\end{eqnarray*}
Because $J$ is independent of $Z$, by \Cref{lem:tree-is-martingale} we have $\tilde{S}^0(J,Z), \tilde{S}^1(J,Z), \ldots, \tilde{S}^H(J,Z)$ is a martingale with differences $|\Delta_h| \leq a$. This yields $\E_J[\Delta_h^2] \le a^2$, and by \Cref{lem:martingale-variance}

\[ \E_J[\tilde{S}(J, Z)^2 ]  \leq a^2H.\]

\begin{figure}[!htbp]
    \centering
\makebox[\linewidth][c]{
        \hspace*{-1.9em}
        \resizebox{1.07\linewidth}{!}{\input{lemmafig.tex}}
    }    \vspace{-20pt}
\caption{Examples of 4-ary tree weightings present in the proof of \Cref{lem:accuracy-bound} to build intuition. In this specific example, $H = 2, k=2, m =16, a = 1.1, \phi = \mean$.
\newline\hspace*{1.5em}
The first tree represents the uncentered weighting $\tilde{w}$. We have a fixed input $X \in \{0,1\}^{16}$ and mechanism output $W = M(E_2(X))$ presented. From those, we compute the unshifted mechanism difference $Z$ and the uncentered weights $\tilde{w}(u,u_i)$ which are recorded on the tree edges. These are computed by $\pm \Clip_{[-a,a]}(d_u(Z))$ for the appropriate alternating differences $d^{02}_u(Z) = \mu_{u_2}(Z) - \mu_{u_0}(Z)$ for edges to $u_1, u_3$ and $d^{13}_u(Z) = \mu_{u_3}(Z) - \mu_{u_1}(Z)$ for edges to $u_0, u_2.$ We observe that for each interior node, the sum of the weights to the children is 0 so when $J \sim \Uniform([16])$, the running edge sum $\tilde{S}^h(J,Z)$ is a martingale.
\newline\hspace*{1.5em}
For the middle tree, $J = 15$ and we construct the centered weighting $w^{+1}$ because $s = 2X_J - 1.$ Flipping $X_{15}$ decreases the prefix of the blocks after $J$ by 1, so the shifted residual goes up by 1.0 (in this diagram, only two entries 1.1 and 1.9 in green). Along the orange path to $J$, the weights if there were no clipping would all increase by 0.5 from $s \cdot \tilde{S}(J,Z)$ (so 0.4 to 0.9 and 0.7 to 1.2). However, because $a = 1.1$ the second edge is clipped to 1.1, indicated with a star. Here $\ClipLoss = 0.1.$
\newline\hspace*{1.5em}
For the bottom tree, $J=10$ and we construct the weighting $w^{-1}$ because $s = 2X_J-1=-1.$ All the shifted residuals after $J$ go down by 1 (marked in green). We compare to $s \cdot \tilde{S}(J,Z)$ and the edge weights all increase by $0.5$ ($-0.2$ to $0.3$, $0.5$ to $1.0$). Note the sign flip compared to the edges in the first tree because $s=-1.$ Here there is no clipping so $\ClipLoss=0.$
\newline\hspace*{1.5em}
The main takeaway is that $\tilde{S}(J,Z)$ is the last stage of a martingale so it is centered on 0. However, $S_s(J,Y)$ would shift everything by $H \cdot \frac{1}{2}$ if not for the clipping. Clipping only occurs when the mechanism error is large. The bulk of \Cref{lem:accuracy-bound} is directly connecting $\E[\ClipLoss]$ with the expected error of the mechanism.
} \label{fig:lemma-fig}
\end{figure}

Now we compare $\tilde{S}(J,Z)$ to the centered scores $S(J, Y).$ Firstly, set $s = 2X_{J} - 1 \in \{\pm 1\}$ which keeps track if flipping coordinate $J$ adds or subtracts 1. In particular, we will consider the difference $S_s(J,Y) - s \cdot \tilde{S}(J,Z).$ Observe that $Y-Z$ is 0 on blocks before $J$ and $s$ on blocks after $J.$ See \Cref{fig:lemma-fig} for a worked example for intuition.

At an interior node $u$, let us compare $w^s(u, u_i)$ and $s \cdot \tilde{w}(u, u_i)$ across the four cases $J \in u_i$. Using the equivariance of $\phi$, $s^2 = 1$, and the oddness $-\Clip_{[-a,a]}(x) = \Clip_{[-a,a]}(-x)$, we obtain:

\begin{center}
\small
\setlength{\tabcolsep}{5pt}
\begin{tabular}{@{}llll@{}}
\toprule
Case & Probes vs.\ $J$ & Active probe on $Y$ & $w^s(u, u_i) - s \cdot \tilde{w}(u, u_i)$ \\
\midrule
$J \in u_0$ & After $J$ & $d_u^{13}(Y) = d_u^{13}(Z)$ & $\Clip_{[-a,a]}(-s \cdot d_u^{13}(Z) + \tfrac{1}{2}) - \Clip_{[-a,a]}(-s \cdot d_u^{13}(Z))$ \\[3pt]
$J \in u_1$ & Surrounding $J$ & $d_u^{02}(Y) = d_u^{02}(Z) + s$ & $\Clip_{[-a,a]}(\phantom{-}s \cdot d_u^{02}(Z) + \tfrac{1}{2}) - \Clip_{[-a,a]}(\phantom{-}s \cdot d_u^{02}(Z))$ \\[3pt]
$J \in u_2$ & Surrounding $J$ & $d_u^{13}(Y) = d_u^{13}(Z) + s$ & $\Clip_{[-a,a]}(\phantom{-}s \cdot d_u^{13}(Z) + \tfrac{1}{2}) - \Clip_{[-a,a]}(\phantom{-}s \cdot d_u^{13}(Z))$ \\[3pt]
$J \in u_3$ & Before $J$ & $d_u^{02}(Y) = d_u^{02}(Z)$ & $\Clip_{[-a,a]}(-s \cdot d_u^{02}(Z) + \tfrac{1}{2}) - \Clip_{[-a,a]}(-s \cdot d_u^{02}(Z))$ \\
\bottomrule
\end{tabular}
\end{center}

In all four cases, letting $d_u \in \{d_u^{13}, d_u^{02}\}$ denote the active probe difference on $u_i$, the difference $w^s(u, u_i) - s \cdot \tilde{w}(u, u_i)$ lies in $[0, \tfrac{1}{2}]$ and equals $\tfrac{1}{2}$ whenever $|d_u(Z)| \le a - \tfrac{1}{2}$ (i.e., $d_u(Z) \in [-a + \tfrac{1}{2}, a - \tfrac{1}{2}]$).

Define $\ClipLoss_u = \frac{1}{2} - (w^s(u,u_i)-s\cdot \tilde{w}(u,u_i))$ so that $\ClipLoss_u \in [0,\frac{1}{2}]$ and $\ClipLoss_u \leq \frac{1}{2} \mathds{1} \{ |d_u(Z)| > a- \frac{1}{2}\}.$ Let $\ClipLoss = \sum_{(u,u_i) \in \mathrm{path}(J)}\ClipLoss_u \geq 0.$ Summing up over $(u,u_i)\in \mathrm{path}(J)$ gives 
\[S_{s}(J,Y) - s \cdot \tilde{S}(J,Z) = \frac{H}{2} - \ClipLoss.\]
By $J \sim \Uniform({[m])}$ and since a node $u$ is visited by $J$ with probability $|u|/m$, we get
\begin{align*}
   \E_{J}[\ClipLoss] = \frac{1}{m} \sum_{j \in [m]} \sum_{(u,u_i) \in \mathrm{path}(j) }\ClipLoss_u
   = \sum_{u \in \mathcal{T}_{\mathrm{int}}} \frac{|u|}{m} \E_J[\ClipLoss_u| J \in u]
\end{align*}

For $\phi = \midr$, we use $\mathds{1}{\{|d| > a-1\}} \le \frac{|d|}{(a-1)}$ for any $d \in \{d^{02},d^{13}\}$ to get
\begin{align}
\E_J[\ClipLoss_{u, \midr} \mid J \in u] &\le
\frac{1}{4} \mathds{1}\{{|d_{u,\midr}^{02}| > a-\frac{1}{2}}\} + \frac{1}{4} \mathds{1}\{{|d_{u, \midr}^{13}| > a-\frac{1}{2}}\}\nonumber\\
&\le \frac{1}{4} \mathds{1}\{{|d_{u,\midr}^{02}| > a-1}\} + \frac{1}{4} \mathds{1}\{{|d_{u, \midr}^{13}| > a-1}\}\nonumber\\ &\leq  \left(  \frac{1}{4(a-1)} \right) \left( |d_{u, \midr}^{02}| + |d_{u, \midr}^{13}| \right)\nonumber \\ 
&=  \left(  \frac{1}{4(a-1)} \right) \left( |c_{u_2}(Z) - c_{u_0}(Z)| + |c_{u_3}(Z) - c_{u_1}(Z)| \right) \nonumber\\
&\le \left(  \frac{1}{4(a-1)} \right) \sum_{i=0}^3 |c_{u_i}(Z) - c_u(Z)| \qquad \text{ by triangle inequality}\nonumber\\
&\le  \frac{1}{2(a-1)}\left(r_u(Z) - \frac{1}{4}\sum_{i=0}^3 r_{u_i}(Z)\right) \qquad \text{ by \Cref{lem:midrange-stability}.}\label{eq:pu-mid}
\end{align}
Similarly for $\phi = \mean$, we use $\mathds{1}{\{|d| > a-1\}} \le \frac{1}{(a-1)^2} d^2$ and the inequality $(a - b)^2 \le 2(a - c)^2 + 2(b - c)^2$ with $c = \mu_u(Z)$ to yield
\begin{align}
\E_J[\ClipLoss_{u, \mean} \mid J \in u] &\le \left(\frac{1}{4(a-1)^2}\right)(|d_{u, \mean}^{02}|^2 + |d_{u, \mean}^{13}|^2) \nonumber \\
&\le \left(\frac{1}{2(a-1)^2}\right) \sum_{i=0}^3 (\mu_{u_i}(Z) - \mu_u(Z))^2 \nonumber \\
&\leq \left(\frac{2 }{(a-1)^2}\right)\left(V_u(Z) - \frac{1}{4}\sum_{i=0}^3 V_{u_i}(Z)\right) \quad \text{\Cref{lem:variance-decomposition}.}\label{eq:pu-mean} 
\end{align}

We can now bound the expectation of $\ClipLoss$. Let $C_\phi = \frac{1}{2(a-1)}$ if $\phi = \midr$ and $\frac{2}{(a-1)^2}$ if $\phi = \mean$. Since a node $u$ is visited by $J$ with probability $|u|/m$, the expected $\ClipLoss_\phi$ on $J$ is given by (for either potential $P \in \{r, V\}$),

\begin{align*}
\E_{J}[\ClipLoss] &\leq \sum_{u \in \mathcal{T}_{\mathrm{int}}} \frac{|u|}{m}  \left(\frac{1}{4} \mathds{1}\{{|d_{u}^{02}| > a-1}\} + \frac{1}{4} \mathds{1}\{{|d_{u}^{13}| > a-1}\}\right)\\ 
&\le C_\phi \sum_{u \in \mathcal{T}_{\mathrm{int}}} \left( \frac{|u|}{m} P_u(Z) - \sum_{i=0}^3 \frac{|u_i|}{m} P_{u_i}(Z) \right) && \text{\Cref{eq:pu-mid,eq:pu-mean}}\\
&= C_\phi \left( P_{[m]}(Z) - \sum_{\ell \text{ leaf}} \frac{|\ell|}{m} P_\ell(Z) \right) && \text{telescoping sum} \\
&\le C_\phi P_{[m]}(Z). && \text{since } P_{\ell}(Z) \ge 0 
\end{align*}

The telescoping sum step follows from the fact that, besides the root and the leaves, a term associated with any other node is added and subtracted exactly once. Finally, at the root $I_{[m]} = [q]$, we have $r_{[q]}(Z) \le 2 \|Z\|_\infty \le \frac{2R}{k}$ and $V_{[q]}(Z) = \frac{1}{q} \sum_{t \in [q]} (Z_t - \mu_{[q]})^2 \le \frac{1}{q}\sum_{t \in [q]} Z_t^2 = \frac{Q}{k^2}$. 
Therefore
\[\E_J[\ClipLoss_{\midr}] \leq \frac{1}{2(a-1)} \cdot \frac{2R}{k}, \quad \E_J[\ClipLoss_{\mean}] \leq \frac{2}{(a-1)^2}\cdot \frac{Q}{k^2}.\]

Note that $S(J,Y) \geq H/4$ whenever $\ClipLoss + |\tilde{S}(J,Z)| \leq H/4.$ Additionally, \[\E_J[|\tilde{S}(J,Z)|] \leq \sqrt{\E_J[\tilde{S}(J,Z)^2]} \leq a \sqrt{H}. \] Finally, we use Markov's inequality to show
\begin{eqnarray*}
    \Pr_J[S(J,Y) \geq H/4] &\geq& \Pr_J[S_{s}(J,Y) \geq H/4] \\
    &\geq& \Pr_J[|\tilde{S}(J,Z)| + \ClipLoss \leq H/4] \\
    &=& 1- \Pr_J[|\tilde{S}(J,Z)| + \ClipLoss > H/4] \\
    &\geq& 1- \frac{4}{H} \E_J[|\tilde{S}(J,Z)|] - \frac{4}{H} \E_J[\ClipLoss] \\
    &\geq& 1- \frac{4}{H} a \sqrt{H} - \frac{4}{H} \E_J[\ClipLoss]. \\
\end{eqnarray*}
Substituting in our calculation for $\E_J[\ClipLoss_{\midr}]$ and $\E_J[\ClipLoss_{\mean}]$ completes the proof.
\end{proof}

\begin{lemma}[Privacy upper bound]
\label{lem:privacy-bound}
Let $n, H, k$ be positive integers with $k 4^H \le n$, and $M \colon \{0,1\}^n \to \R^n$ be an $(\eps, \delta)$-differentially private mechanism. Under the same distributions for $X, J$, and the residual $Y$ defined in \Cref{lem:accuracy-bound}, we have for any $a > 1$ and $\phi \in \{\mean, \midr\}$:
\[
    \Pr_{X,J,M}[S_\phi(J,Y) \geq H/4] \le 2e^{k\eps-\frac{H}{32a^2}}+ \delta_k.
\]
\end{lemma}

\begin{proof}
Let $W'$ be a fresh run of $M(E_k(X'))$ with residual $Y' = \frac{(W' - A_n E_k(X'))_{[q]}}{k}$. For fixed $X=x,J=j$, the encoded inputs $E_k(x),E_k(x')$
differ in $k$ coordinates. By group privacy and averaging over $X,J$:
\[
\begin{aligned}
\Pr_{X,J,M}[S_\phi(J,Y) \geq H/4 ]
&\le e^{k\eps} \Pr_{X,J,M}[S_\phi(J,Y') \geq H/4]+\delta_k
\leq e^{k\eps}
  \left(2 e^{-\frac{H}{32a^2}} \right)+\delta_k.
\end{aligned}
\]

For the final inequality, applying \Cref{lem:bit-flip-iid} to the function $f(x') = M(E_k(x')) - A_n E_k(x')$ shows that $Y'$ is independent of $J$. For fixed $Y' = y$, $S_{\phi,\tau}(J,Y')$ is step $H$ of a martingale where each step has magnitude at most $a$. By \Cref{lem:azuma}

\[
\Pr_J\left[S_{\phi,\tau}(J, y) \geq H/4 \right] \leq e^{-\frac{(H/4)^2}{2Ha^2}}.
\]
Because this inequality holds for all $y$, we have $\Pr_{X,J,M}[S_{\phi, \tau}(J,Y') \geq H/4] \leq e^{-\frac{H}{32a^2}} $. Using the union bound for $\tau \in \{\pm 1\}$ gives the factor of 2, completing the lemma.
\end{proof}

\subsection{Main result}
Finally, we prove our main theorem by selecting probe rule $\phi = \midr$ for $\alpha_\infty$ and $\phi = \mean$ for $\MeanSE$.

\begin{proof}[Proof of \Cref{thm:lower_bound}]

Along with the variables
\begin{align*}
   & \eta = \max\{\eps,\delta\}, & L = \log(n \eta), && U = 1 + \min\{L,\log(\eta/\delta)\}.
\end{align*}
To get the desired bounds, we set $\beta = 2^{-13}$ and 
\begin{align*}
   k = \max \left\{1, \left\lfloor \frac{\beta \cdot U}{\eta}\right\rfloor \right\} = \Theta(U/\eta), && H = \lfloor \log_4(n/k) \rfloor = \Theta(L), & & a = \frac{1}{16}\sqrt{\frac{H}{1+k\eps}}=\Theta\left(\sqrt{L/U}\right).
\end{align*}
To build some understanding for these variables: $L = \log(n \eta)$ is the tree height after grouping elements into blocks of size $\frac{1}{\eta}.$ We choose the block size $k$ to keep the accumulated privacy loss $\delta_k\le 1/4$, while retaining enough blocks to achieve the tree height required for the lower bound. Then $a$ is chosen as large as possible while still maintaining that $2e^{k\eps - \frac{H}{32a^2}}$ is sufficiently small.

If we let $\Loss_\phi$ be $\frac{4\E_{X,M}[R]}{kH(a-1)}$ for $\phi = \midr$ and $\frac{8\E_{X,M}[Q]}{k^2H(a-1)^2}$ for $\phi = \mean,$ then combining \Cref{lem:accuracy-bound} and \Cref{lem:privacy-bound} gives:
\[ 1 - \frac{4a}{\sqrt{H}} - \Loss_\phi \leq \Pr_{X,J,M}[S_\phi(J,Y) \geq H/4] \leq 2e^{k\eps - \frac{H}{32a^2}} + \delta_k. \]
Rearranging yields
\[ \Loss_\phi \geq 1 - \frac{4a}{\sqrt{H}} - 2 e^{k\eps - \frac{H}{32a^2}} - \delta_k.\]
Note that all these terms are small enough that $\Loss_\phi \geq \frac{5}{12}.$ We go term by term:
\begin{itemize}
    \item $\frac{4a}{\sqrt{H}} = \frac{1}{4\sqrt{1+k\eps}} \leq \frac{1}{4}$
    \item $2e^{k\eps - \frac{H}{32a^2}} = 2 e^{k\eps - 8(1+k\eps)} \leq 2e^{-8}$
    \item $\delta_k = \delta \sum_{i=0}^{k-1} e^{i\eps} \leq \delta k e^{k\eps}:$ If $k=1$ then this is equal to $\delta \leq \frac{1}{4}.$ Otherwise
\begin{align*}
\delta_k &\leq \delta k e^{k\eps}\\
&\leq \delta \cdot \frac{\beta U}{\eta} \cdot e^{k\eps} \\
&\leq \delta \cdot \frac{\beta U}{\eta} \cdot e^{\beta U} & \textrm{because } k \eps \leq k \eta = \eta \left\lfloor \frac{\beta \cdot U}{\eta}\right\rfloor \le \beta U\\
&\leq \beta U e^{1-(1-\beta)U} & \textrm{because } U \leq 1 + \log(\eta/\delta) \textrm{ means } \frac{\delta}{\eta} \leq e^{1-U} \\
&\leq \frac{\beta}{1-\beta} & \textrm{because } xe^{1-cx} \textrm{ has a maximum of } \frac{1}{c} \\
&\leq 1/4.
\end{align*}
\end{itemize}
Putting these together we have $\Loss_\phi \geq 1 - \frac{1}{4} - 2e^{-8} - \frac{1}{4} > \frac{5}{12}.$ Because we observe $\Loss_\phi$ is bounded below by a positive constant, we can find the asymptotics of $\E_{X,M}[R]$ and $\E_{X,M}[Q].$ Note the following inequalities
\begin{itemize}
\item $U \le 1 + L$
\item$ k \geq \frac{\beta U}{2 \eta}.$ 
\item $1 + k \eps \leq 2 + \beta U \le 3U$.
\item $\frac{q}{n} = \frac{k 4 ^{\lfloor \log_4(n/k) \rfloor}}{n} > \frac{1}{4}.$
\item $H \geq \log_4(n/k)-1 \geq \log_4(\frac{n\eta}{1+\beta U}) - 1 \geq \frac{L-\log(1+\beta U)}{\log 4}-1
  \geq \frac{L}{2}
  \geq 2048+\frac{1+L}{8}
  \geq 1024(1+k\eps)$, for $L\geq 5500$ and so $n\eta \geq C_0:= e^{5500}$.
\item $a \geq 2 \implies a-1 \geq \frac{a}{2}
    = \frac{1}{32}\sqrt{\frac{H}{1+k\eps}}.$ 
\end{itemize}

Now we substitute $\Loss_\phi$ and rearrange. For the $\alpha_\infty$ bound:
\begin{eqnarray*}
    \alpha_\infty(M) \geq \E_{X,M}[R] 
    \geq \frac{5}{48} k H(a-1) 
    \geq \frac{5}{1536} \frac{kH^{3/2}}{\sqrt{1+k\eps}} 
    \geq& \frac{5}{1536} \frac{(\beta U / (2\eta)) (L/2)^{3/2}}{\sqrt{3U}} 
    \geq c \frac{L^{3/2} \sqrt{U}}{\eta}.
\end{eqnarray*}

Now we handle $\MeanSE$ and $\MaxSE$ bounds. Since the full-stream mean squared error is at least $q/n$
times the mean squared error on the active prefix, and our construction guarantees $q \ge n / 4$ we find:
\begin{align*}
    \MaxSE(M) \geq \MeanSE(M) 
    &\geq \frac{q}{n} \E_{X,M}[Q] \\
    &\geq \frac{q}{n} \cdot \frac{5}{96} k^2 H(a-1)^2 \\
    &\geq \frac{5}{393216} \frac{k^2 H^2}{1+k\eps} \\
    &\geq \frac{5}{393216} \left(\frac{\beta U}{2\eta}\right)^2 \cdot (L/2)^2 / (3U) 
    \geq c \frac{L^2U}{\eta^2}.\qedhere
\end{align*}
where the both bounds hold with $c = \frac{5\beta^2}{393216 \cdot 48} = \frac{5}{9 \cdot 2^{47}}$. 
\end{proof}

\section{Acknowledgments}
We thank Jalaj Upadhyay for helpful feedback on an earlier draft, including  \cite{upadhyay2026exposition}.

\paragraph{AI Methodology.} We used a combination of ChatGPT 6 Astra and multiple Gemini models (Argon, Carbon, 3.8 Flash) to help discover key proof ingredients and aid in drafting, figure generating, literature search, reviewing and editing. These tools were used interactively starting from the initial solution of Bairaktari and Larsen \cite{bairaktari2026binary-orig}, and materially impacted all sections of the paper. The human authors spent significant time working collaboratively with the AI tools to make the arguments as concise and general as possible. We also improved the presentation, and verified the correctness and originality of all content, including references.

\bibliographystyle{alpha}
\bibliography{ref}

\appendix
\section{Remaining upper and lower bounds}\label{app:matching-upper-bounds}

In this appendix, we prove the remaining claims referenced in \Cref{tab:bounds}. \Cref{thm:lower_bound} proves the lower bound for $n \cdot \max\{\eps, \delta \} \geq C_0$ which covers the first three rows. We first prove the lower bound for the last row: 

\subsection{Lower bound for high privacy regime}

\begin{lemma}[Lower bound when $n \cdot \max\{\eps, \delta \} < C_0$]
    For any fixed $C_0 \geq 1$ and $n \geq 8C_0,$ whenever $n \max \{ \eps, \delta\} < C_0$, $0 \le \eps \le 1$ and $0 \le \delta \le 1/4$, if $M$ is an $(\eps, \delta)$-DP continual counting mechanism, then $\alpha_\infty(M) = \Omega(n)$ and $\MaxSE(M) \geq \MeanSE(M) = \Omega(n^2).$
\end{lemma}
\begin{proof}
Pick $1\le k\le\lfloor n/2\rfloor$ satisfying $k\varepsilon\le1/8$ and
$\delta_k:=\delta\sum_{j=0}^{k-1}e^{j\varepsilon}\le1/4$ (e.g., 
$k=\left\lfloor\frac{n}{8C_0}\right\rfloor$).
Set $x^{(0)}=0^n$, $x^{(1)}=1^k0^{n-k}$, and
\[
T=\{\lfloor n/2\rfloor+1,\ldots,n\},
\qquad m=|T|=\lceil n/2\rceil.
\]
The true prefix sums on $T$ equal $bk$ for input $x^{(b)}$.
Let $P_b$ denote the distribution of $M(x^{(b)})$, and define
\[
E=\left\{w:\bar w_T:=\frac1m\sum_{i\in T}w_i\ge k/2\right\}.
\]
\Cref{lem:group-privacy} gives $1-P_1(E^c)\le e^{k\varepsilon}P_0(E)+\delta_k$,
hence
\[
\max\{P_0(E),P_1(E^c)\}
\ge \frac{1-\delta_k}{1+e^{k\varepsilon}}
\ge \frac{3}{4(1+e^{1/8})}
=:p.
\]
Thus, for some $b\in\{0,1\}$, with probability at least $p$,
$|\bar w_T-bk|\ge k/2$.
On this event, writing $z=w-A_nx^{(b)}$, we have
\[
\|z\|_\infty
\ge \left|\frac1m\sum_{i\in T}z_i\right|
\ge \frac{k}{2}.
\]
Moreover, Cauchy--Schwarz gives
\[
\frac1n\|z\|_2^2
\ge \frac1n\sum_{i\in T}z_i^2
\ge \frac{1}{nm}\left(\sum_{i\in T}z_i\right)^2
= \frac mn|\bar w_T-bk|^2
\ge \frac mn\frac{k^2}{4}
\ge \frac{k^2}{8}.
\]
Taking expectations yields
\[
\sup_x\mathbb{E}\|M(x)-A_nx\|_\infty\ge \frac{pk}{2},
\qquad
\sup_x\mathbb{E}\!\left[\frac1n\|M(x)-A_nx\|_2^2\right]
\ge \frac{pk^2}{8}.
\]
\end{proof}

\subsection{Blocked binary tree mechanism matches lower bound}
First we define the relevant noise distributions that are common in differential privacy and when they satisfy pure or approximate DP:
\begin{definition}
    For scale parameter $b > 0$, the \define{Laplace distribution} $\Lap(b)$ has probability density $f(z) = \frac{1}{2b} e^{-|z|/b}$ for all $z\in\R.$
\end{definition}
\begin{lemma}[\cite{dworkrothbook} Theorem 3.6]\label{lem:laplace-private}
    For any $\eps > 0$ and any query function $f: \{0,1\}^n \to \R^k$ with $\ell_1$-sensitivity $\Delta_1$, the mechanism $f(x) + (\xi_1,\ldots,\xi_k)$ where $\xi_i \sim \Lap(\Delta_1/\eps)$  independently satisfies $(\eps,0)$-DP.
\end{lemma}
\begin{definition}
    For variance parameter $\sigma^2>0,$ the \define{Gaussian distribution} $\mathcal{N}(0,\sigma^2)$ has probability density $f(z) = \frac{1}{\sqrt{2\pi}\sigma}e^{-z^2/(2\sigma^2)}.$
\end{definition}
\begin{lemma}[\cite{balle2018improvinggaussianmechanismdifferential} Theorem 8]\label{lem:gauss-private}
    For any $\eps \ge 0$, $0 \leq \delta \leq 1$ and any query function $f: \{0,1\}^n \to \R^k$ with $\ell_2$-sensitivity $\Delta_2$, the mechanism $f(x) + (\xi_1, \ldots,\xi_k)$ where $\xi_i \sim \mathcal{N}(0,\sigma^2)$ independently satisfies $(\eps,\delta)$-DP if and only if 
    \[\delta_G(\sigma) := \Phi\!\left(\frac{\Delta_2}{2\sigma} - \frac{\eps\sigma}{\Delta_2}\right) - e^\eps\Phi\!\left(-\frac{\Delta_2}{2\sigma} - \frac{\eps\sigma}{\Delta_2}\right) \le \delta\]
    where $\Phi(t) = \Pr[\mathcal{N}(0,1) \leq t] = \frac{1}{\sqrt{2\pi}} \int_{-\infty}^t e^{-y^2/2} dy = \int_{-\infty}^t \varphi(t) dt$ and $\varphi(t) = \frac{1}{\sqrt{2\pi}}e^{-t^2/2}.$ 
\end{lemma}

Additionally, we recall the classic Binary Tree Mechanism from \cite{dworkcontinual} and \cite{chan-continual-release} and its blocked variant \cite{jain-price-continual, ghazi2022privatecountingdistinctkoccurring} which aggregates the stream into blocks of size $B$ before applying the Binary Tree Mechanism. The values on each block become constant.

\begin{algorithm}[H]
    \caption{Binary Tree Mechanism (\textsc{BinaryTreeMechanism}) \cite{dworkcontinual,chan-continual-release}}\label{alg:btm}
    \begin{algorithmic}[1]
    \Require Stream $s = (s_1, \dots, s_m) \in \R^m$, noise distribution $\mathcal{D}$.
    \Ensure Noisy prefix sums $\widehat{P} = (\widehat{P}_1, \dots, \widehat{P}_m) \in \R^m$.
    \For{$\ell = 0$ \textbf{to} $\lfloor \log_2 m \rfloor$ \textbf{and} $k = 1$ \textbf{to} $\lfloor m / 2^\ell \rfloor$}
        \State $I_{\ell, k} \gets [(k-1)2^\ell + 1, \, k 2^\ell]$
        \State Sample $\xi_{\ell, k} \sim \mathcal{D}$ independently and set $\hat{c}_{\ell, k} \gets \sum_{i \in I_{\ell, k}} s_i + \xi_{\ell, k}$
    \EndFor
    \For{$j = 1$ \textbf{to} $m$}
        \State Write $j = \sum_{\ell=0}^{\lfloor \log_2 j \rfloor} b_\ell 2^\ell$ in binary with $b_\ell \in \{0,1\}$
        \State $\widehat{P}_j \gets \sum_{\ell : b_\ell = 1} \hat{c}_{\ell, \lfloor j / 2^\ell \rfloor}$
    \EndFor
    \State \Return $\widehat{P} = (\widehat{P}_1, \dots, \widehat{P}_m)$
    \end{algorithmic}
\end{algorithm}
\begin{algorithm}[H]
    \caption{Blocked Binary Tree Mechanism (\textsc{BlockedBTM}) \cite{jain-price-continual,ghazi2022privatecountingdistinctkoccurring}}\label{alg:blocked_btm}
    \begin{algorithmic}[1]
    \Require Bitstream $x \in \{0,1\}^n$, block size $B \in [n]$, noise distribution $\mathcal{D}$.
    \Ensure Prefix sum estimates $\hat{p} = (\hat{p}_1, \dots, \hat{p}_n) \in \R^n$.
    \State $m \gets \lfloor n / B \rfloor$
    \For{$j = 1$ \textbf{to} $m$}
        \State $S_j \gets \sum_{t=(j-1)B+1}^{jB} x_t$
    \EndFor
    \State $\widehat{P} \gets$ \Call{BinaryTreeMechanism}{$(S_1, \dots, S_m), \mathcal{D}$}
    \State $\widehat{P}_0 \gets 0$
    \For{$t = 1$ \textbf{to} $n$}
        \State $\hat{p}_t \gets \widehat{P}_{\lfloor t / B \rfloor}$
    \EndFor
    \State \Return $\hat{p} = (\hat{p}_1, \dots, \hat{p}_n)$
    \end{algorithmic}
\end{algorithm}

\begin{lemma}[Privacy of \textsc{BlockedBTM} with Laplace noise]\label{lem:blocked-lap-dp}
    For any block size $B \in [n]$ with $m = \lfloor n/B \rfloor \ge 1$, let $h = 1 + \lfloor \log_2 m \rfloor$. For any $\eps > 0$, \Cref{alg:blocked_btm} instantiated with $\mathcal{D} = \Lap(h/\eps)$ is $(\eps, 0)$-DP.
\end{lemma}
\begin{proof}
    Changing a single bit of $x \in \{0,1\}^n$ changes at most one block sum $S_j$ by at most $1$. In \Cref{alg:btm}, each block $j \in [m]$ belongs to at most $h = 1 + \lfloor \log_2 m \rfloor$ intervals in total. Thus, the vector $(\sum_{i\in I_{\ell,k}} s_i)_{\ell,k}$ has $\ell_1$-sensitivity $\Delta_1 \le h$. By \Cref{lem:laplace-private} and by post-processing, \Cref{alg:blocked_btm} with $\mathcal{D} = \Lap(h/\eps)$ is $(\eps, 0)$-DP.
\end{proof}
\begin{lemma}[Privacy of \textsc{BlockedBTM} with Gaussian noise]\label{lem:blocked-gauss-dp}
    For any block size $B \in [n]$ with $m = \lfloor n/B \rfloor \ge 1$, let $h = 1 + \lfloor \log_2 m \rfloor$. For any $0 \le \eps \le 1$ and $0 < \delta \le 1/4$, let $\eta = \max\{\eps, \delta\}$ and $V = 1 + \log(\eta/\delta)$. If $\sigma = \frac{2\sqrt{hV}}{\eta}$, then \Cref{alg:blocked_btm} instantiated with $\mathcal{D} = \mathcal{N}(0, \sigma^2)$ is $(\eps, \delta)$-DP.
\end{lemma}
\begin{proof}
    Identically to the case of \Cref{lem:blocked-lap-dp}, changing $x$ changes at most $h$ intervals by at most $1$, so the vector of interval sums has $\ell_2$-sensitivity $\Delta_2 \le \sqrt{h}$. By \Cref{lem:gauss-private} and post-processing, it suffices to verify that $\delta_G(\sigma) \le \delta$ for $\Delta_2 = \sqrt{h}$. We consider two cases.
    
    \textbf{Case 1: $\delta \le \eps$.} Here $\eta = \eps$ and $V = 1 + \log(\eps/\delta) \ge 1$. Let $\rho = \frac{\Delta_2}{\sigma} = \frac{\eps}{2\sqrt{V}}$ and $t = \frac{\eps}{\rho} - \frac{\rho}{2} = 2\sqrt{V} - \frac{\eps}{4\sqrt{V}}$. Since $\eps \le 1$ and $V \ge 1$, we have $t \ge \frac{7}{4}\sqrt{V} > \sqrt{2V} > 0$. Using $1 - e^{-x} \le x$ and $z - t \le z$ for $z \ge t > 0$, we obtain
    \begin{align*}
        \delta_G(\sigma)
        &= \Phi(-t) - e^\eps\Phi(-t-\rho) \\
        &= \int_t^\infty \varphi(z)\,dz - e^\eps\int_t^\infty \varphi(z+\rho)\,dz \\
        &= \int_t^\infty \bigl(1 - e^{-\rho(z-t)}\bigr)\varphi(z)\,dz \\
        &\le \rho\int_t^\infty (z-t)\varphi(z)\,dz \\
        &\le \rho\int_t^\infty z\varphi(z)\,dz
        = \rho\int_t^\infty -\varphi'(z)\,dz
        = \rho\varphi(t).
    \end{align*}
    Since $t > \sqrt{2V}$, we have $\varphi(t) = \frac{1}{\sqrt{2\pi}}e^{-t^2/2} < e^{-V}$. Combined with $\rho = \frac{\eps}{2\sqrt{V}} < \eps$, this yields
    \[
        \delta_G(\sigma) \le \rho\varphi(t) < \eps e^{-V} = \eps e^{-1 - \log(\eps/\delta)} = \frac{\delta}{e} < \delta.
    \]
    \textbf{Case 2: $\delta > \eps$.} Here $\eta = \delta$, $V = 1$, and $\sigma = 2\sqrt{h}/\delta$, so $\frac{\Delta_2}{\sigma} = \frac{\delta}{2}$. Because any $(0,\delta)$-DP mechanism is automatically $(\eps,\delta)$-DP, we can assume $\eps = 0$, and bound $\delta_G(\sigma)$
    \[
        \delta_G(\sigma) \le \Phi\!\left(\frac{\Delta_2}{2\sigma}\right) - \Phi\!\left(-\frac{\Delta_2}{2\sigma}\right) \le \frac{\Delta_2}{\sigma}\varphi(0) = \frac{\delta}{2\sqrt{2\pi}} < \delta.
    \]
    Thus \Cref{alg:blocked_btm} is $(\eps, \delta)$-DP in both cases.
\end{proof}
Next, we bound the utility of \Cref{alg:blocked_btm} in terms of the block size $B$ and the underlying Binary Tree Mechanism error on $m = \lfloor n/B \rfloor$ blocks.
\begin{lemma}[Error of \text{\textsc{BlockedBTM}}]\label{lem:blocked-utility}
    Fix any block size $B \in [n]$, and let $m = \lfloor n/B \rfloor \ge 1$ and $h = 1 + \lfloor \log_2 m \rfloor$. Then for any noise distribution $\mathcal{D}$ with mean 0, the output $\hat{p}$ of \Cref{alg:blocked_btm} satisfies
    \[\alpha_\infty(\text{\textsc{BlockedBTM}}) \le B + \E\max_{1 \le j \le m}|N_j|\]
    and
    \[\MeanSE(\textsc{BlockedBTM}) \le \MaxSE(\textsc{BlockedBTM}) \le B^2 + \max_{1 \le j \le m}\E[N_j^2]\]
    where $N_j = \widehat{P}_j - \sum_{i=1}^j S_i$ is the total noise in the $j$-th prefix sum of \Cref{alg:btm}.
\end{lemma}
\begin{proof}
    At any time $t \in [n]$, \Cref{alg:blocked_btm} outputs $\hat{p}_t = \widehat{P}_{\lfloor t/B \rfloor}$ (with $N_0 = 0$). Since $(A_n x)_t = \sum_{i=1}^{\lfloor t/B \rfloor} S_i + r_t$, where $r_t = \sum_{i=B\lfloor t/B \rfloor + 1}^t x_i \in [0, B-1]$ is the partial sum of the unfinished block, the error at time $t$ is $\hat{p}_t - (A_n x)_t = N_{\lfloor t/B \rfloor} - r_t$. By the triangle inequality, $|\hat{p}_t - (A_n x)_t| \le B + \max_{1 \le j \le m}|N_j|$, and since $\E[N_j] = 0$, we have $\E[(\hat{p}_t - (A_n x)_t)^2] = \E[N_{\lfloor t/B \rfloor}^2] + r_t^2 \le B^2 + \max_{1 \le j \le m}\E[N_j^2]$.
\end{proof}
\begin{corollary}\label{cor:bounds-blockedbtm}
If $\mathcal{D} = \Lap(b)$, then
\begin{equation}\label{eq:blocked-lap-err}
\alpha_\infty(\text{\textsc{BlockedBTM}}) = O\bigl(B + b h\bigr), \qquad
\MaxSE(\text{\textsc{BlockedBTM}}) = O\bigl(B^2 + b^2 h\bigr).
\end{equation}
If $\mathcal{D} = \mathcal{N}(0, \sigma^2)$, then
\begin{equation}\label{eq:blocked-gauss-err}
\alpha_\infty(\text{\textsc{BlockedBTM}}) = O\bigl(B + \sigma h\bigr), \qquad
\MaxSE(\text{\textsc{BlockedBTM}}) = O\bigl(B^2 + \sigma^2 h\bigr).
\end{equation}
\end{corollary}
\begin{proof}
Each prefix noise $N_j$ is the sum of at most $h = 1 + \lfloor \log_2 m \rfloor$ independent mean-zero draws from $\mathcal{D}$. By \Cref{lem:blocked-utility}, it suffices to bound $\max_{1 \le j \le m} \E[N_j^2]$ (for $\MaxSE$) and $\E[\max_{1 \le j \le m} |N_j|]$ (for $\alpha_\infty$), using $\log(2m) \le (h+1)\log 2 < 2h$.

For the latter, for any $\lambda > 0$, concavity of $\log$ (Jensen's inequality) gives
\[
\E\left[\max_{1 \le j \le m} |N_j|\right]
= \frac{1}{\lambda}\E\left[\log \max_{1 \le j \le m} e^{\lambda |N_j|}\right]
\le \frac{1}{\lambda}\E\left[\log \sum_{j=1}^m \bigl(e^{\lambda N_j} + e^{-\lambda N_j}\bigr)\right]
\le \frac{1}{\lambda}\log \sum_{j=1}^m \E\bigl[e^{\lambda N_j} + e^{-\lambda N_j}\bigr].
\]
\begin{itemize}
    \item \textbf{Laplace ($\mathcal{D} = \Lap(b)$):} Each draw $\xi \sim \Lap(b)$ satisfies $\E[\xi^2] = 2b^2$ and $\E[e^{\pm \xi/(2b)}] = \frac{1}{1 - 1/4} = \frac{4}{3}$. By independence of the summands in $N_j$, for every $j \in [m]$ we have $\E[N_j^2] \le 2b^2 h$ (so $\max_{1 \le j \le m} \E[N_j^2] = O(b^2 h)$) and $\E[e^{\pm N_j/(2b)}] \le (4/3)^h$. Setting $\lambda = \frac{1}{2b}$ above gives
    \[
    \E\left[\max_{1 \le j \le m} |N_j|\right] \le 2b \log\bigl(2m(4/3)^h\bigr) = 2b\bigl(\log(2m) + h\log(4/3)\bigr) = O(bh).
    \]
    \item \textbf{Gaussian ($\mathcal{D} = \mathcal{N}(0, \sigma^2)$):} For every $j \in [m]$, $N_j$ is Gaussian with variance $\E[N_j^2] \le \sigma^2 h$ (so $\max_{1 \le j \le m} \E[N_j^2] = O(\sigma^2 h)$) and $\E[e^{\pm \lambda N_j}] \le e^{\lambda^2 h\sigma^2/2}$ for all $\lambda > 0$. Setting $\lambda = \frac{\sqrt{2\log(2m)}}{\sigma\sqrt{h}}$ above gives
    \[
    \E\left[\max_{1 \le j \le m} |N_j|\right] \le \frac{1}{\lambda}\log\bigl(2m e^{\lambda^2 h\sigma^2/2}\bigr) = \sigma\sqrt{2h\log(2m)} = O(\sigma h). \qedhere
    \]
\end{itemize}
\end{proof}
Combining \Cref{lem:blocked-lap-dp,lem:blocked-gauss-dp} and \Cref{cor:bounds-blockedbtm} immediately yields the matching upper bounds in all four rows of \Cref{tab:bounds}.
\begin{corollary}[Matching upper bounds for \Cref{tab:bounds}]\label{cor:matching-upper-bounds}\label{cor:upper-bound}
    Let $n \ge 1$, $0 \le \eps \le 1$, $0 \le \delta \le 1/4$, and $\eta = \max\{\eps, \delta\}$. If $n\eta \ge C_0$, let $L = \log(n\eta)$ and $U = 1 + \min\{L, \log(\eta/\delta)\}$. Then \Cref{alg:blocked_btm} with block size
    \[
        B = \begin{cases}
            \lceil 1/\eps \rceil & \text{if } \delta \le \eps/2, \\
            \lceil 1/\delta \rceil & \text{if } \delta > \eps/2,
        \end{cases}
    \]
    and with Laplace noise ($\mathcal{D} = \Lap(h/\eps)$) when $\delta \le 1/n$ and Gaussian noise ($\mathcal{D} = \mathcal{N}(0, \sigma^2)$ with $\sigma = \frac{2\sqrt{hV}}{\eta}$) when $\delta > 1/n$, is $(\eps, \delta)$-DP and achieves
    \[
        \alpha_\infty(\textsc{BlockedBTM}) = O\!\left(\frac{L^{3/2}\sqrt{U}}{\eta}\right), \] \[
        \MeanSE(\textsc{BlockedBTM}) \le \MaxSE(\textsc{BlockedBTM}) = O\!\left(\frac{L^2 U}{\eta^2}\right),
    \]
    matching the first three rows of \Cref{tab:bounds}. When $n\eta < C_0$, the mechanism that always outputs $0^n$ is $(0,0)$-DP and achieves $\alpha_\infty \le n = O(n)$ and $\MeanSE \le \MaxSE \le n^2 = O(n^2)$, matching the fourth row of \Cref{tab:bounds}.
\end{corollary}
\begin{proof}
    When $n\eta \ge C_0 \ge 8$, our choice of $B$ satisfies $\frac{1}{\eta} \le B \le \frac{2}{\eta} + 1 \le \frac{3}{\eta}$, so $B = \Theta(1/\eta) \le n/2$, $m = \lfloor n/B \rfloor = \Theta(n\eta) \ge 2$, and $h = 1 + \lfloor \log_2 m \rfloor = \Theta(\log(n\eta)) = \Theta(L)$.
    \begin{itemize}
        \item \textbf{Row 1 ($0 \le \delta \le 1/n$):} Since $n\eta \ge 8$ and $n\delta \le 1$, we have $\eta = \eps \ge 8/n > 2\delta$, so $B = \lceil 1/\eps \rceil = \Theta(1/\eps)$, $L = \log(n\eps)$, and $U = 1 + L = \Theta(L)$. By \Cref{lem:blocked-lap-dp,lem:blocked-utility} and \Cref{cor:bounds-blockedbtm} with $b = h/\eps$, \textsc{BlockedBTM} is $(\eps, 0)$-DP and satisfies $\alpha_\infty = O(B + h^2/\eps) = O(\log^2(n\eps)/\eps)$ and $\MaxSE = O(B^2 + h^3/\eps^2) = O(\log^3(n\eps)/\eps^2)$.
        \item \textbf{Row 2 ($1/n < \delta \le \eps/2$):} Here $\eta = \eps$, $B = \lceil 1/\eps \rceil = \Theta(1/\eps)$, $L = \log(n\eps)$, and $\log(\eta/\delta) < L$, so $U = V = 1 + \log(\eps/\delta) = \Theta(\log(\eps/\delta))$. By \Cref{lem:blocked-gauss-dp,lem:blocked-utility} and \Cref{cor:bounds-blockedbtm} with $\sigma = \frac{2\sqrt{hV}}{\eps}$, \textsc{BlockedBTM} is $(\eps, \delta)$-DP and satisfies $\alpha_\infty = O(B + h^{3/2}\sqrt{V}/\eps) = O(\log^{3/2}(n\eps)\sqrt{\log(\eps/\delta)}/\eps)$ and $\MaxSE = O(B^2 + h^2 V/\eps^2) = O(\log^2(n\eps)\log(\eps/\delta)/\eps^2)$.
        \item \textbf{Row 3 ($\eps/2 < \delta \le 1/4$):} Here $\eta = \Theta(\delta)$, $B = \lceil 1/\delta \rceil = \Theta(1/\delta)$, $h = \Theta(\log(n\delta))$, and $U, V = \Theta(1)$. By \Cref{lem:blocked-gauss-dp,lem:blocked-utility} and \Cref{cor:bounds-blockedbtm} with $\sigma = \frac{2\sqrt{hV}}{\eta} = \Theta(\sqrt{h}/\delta)$, \textsc{BlockedBTM} is $(\eps, \delta)$-DP and satisfies $\alpha_\infty = O(B + h^{3/2}/\delta) = O(\log^{3/2}(n\delta)/\delta)$ and $\MaxSE = O(B^2 + h^2/\delta^2) = O(\log^2(n\delta)/\delta^2)$.
        \item \textbf{Row 4 ($n\eta < C_0$):} Since $x \in \{0,1\}^n$ has $0 \le (A_n x)_t \le n$ for all $t \in [n]$, outputting $0^n$ has $\alpha_\infty \le n$ and $\MeanSE \le \MaxSE \le n^2$. \qedhere
    \end{itemize}
\end{proof}
We lastly present a lemma stating that you can ignore the diversion to Blocked Binary Tree Mechanisms if $\max \{ \eps, \delta\} > \frac{1}{n^{1-c}}$ for some constant $c>0.$
\begin{lemma}\label{lem:vanilla-ok-c}
    If $\eta \ge n^{-(1-c)}$ for a fixed constant $c > 0$, the Binary Tree Mechanism with the same noise selection in \Cref{cor:upper-bound} also achieves the bounds of \Cref{cor:upper-bound}.
\end{lemma}
\begin{proof}
            Setting $B = 1$ gives $m = n$ and $h = 1 + \lfloor \log_2 n \rfloor = \Theta(\log n) = \Theta(L)$ since $c\log n \le L \le \log n$, so~\eqref{eq:blocked-lap-err} and~\eqref{eq:blocked-gauss-err} yield the same asymptotic bounds.
\end{proof}

\section{Constant and high probability lower bounds}
We strengthen \Cref{thm:lower_bound} by giving a lower bound on the probability of simultaneously incurring large $\ell_\infty$ and mean squared error. This probability is bounded below by a positive universal constant throughout the parameter range of \Cref{thm:lower_bound}.

\begin{corollary}
Under the assumptions of \Cref{thm:lower_bound}, every \((\eps,\delta)\)-DP mechanism \(M\) admits an input \(x\) such that, with probability at least \(1-3(2e^{-8-7k\eps}+\delta_k) \),

$$
\|M(x)-A_nx\|_\infty>\frac{kH(a-1)}{64},
\qquad
\frac1n\|M(x)-A_nx\|_2^2>\frac{k^2H(a-1)^2}{512},
$$

where \(k,H,a\) are as in its proof. For fixed \(\eps>0\) and \(\delta\le n^{-1}\), these give \(\Omega(\log^2 n)\) and \(\Omega(\log^3 n)\) orders with failure probability $n^{-\Omega(1)}$.
\end{corollary}
\begin{proof}
Use $X,W,R,Q$ from the proof of Theorem~1. Let $r_* = \frac{kH(a-1)}{64}$ and $q_* =\frac{k^2H(a-1)^2}{128}.$ Then \Cref{lem:accuracy-bound} directly states that for any $X,W$ that have loss $R \leq r_*$
\begin{eqnarray*}
    \Pr_J[S_{\midr}(J,Y) \geq H/4 \mid X, W] &\geq& 1 - \frac{4a}{\sqrt{H}} - \frac{4R}{kH(a-1)} \\
    &\geq& 1 - \frac{1}{4} - \frac{1}{16} \\
    &>& 2/3
\end{eqnarray*}
and identically for $\Pr_J[S_{\mean}(J,Y) \geq H/4 \mid X, W] > 2/3$ whenever $Q \leq q_*.$
Note that \Cref{lem:privacy-bound} directly states that \[\Pr_{X,J,M}[S_\phi(J,Y) \geq H/4] \leq b := 2e^{-8-7k\eps}+\delta_k. \]
Therefore $\Pr_{X,M}[R \leq r_*] \leq \frac{3b}{2}$ and similarly $\Pr_{X,M}[Q \leq q_*] \leq \frac{3b}{2}$.
Because $\frac1n\|M(x)-A_nx\|_2^2 \geq \frac{q}{n} Q > \frac{1}{4} Q,$ both inequalities in the statement hold with probability $1-3b$ for random $X,M$, so there must exist an $x$ by the probabilistic method. 

For fixed $\varepsilon>0$ and $\delta\le n^{-1}$,
we have $k,H=\Theta(\log n)$, $a=\Theta(1)$. In particular because $k\eps = \Theta(U)$, $2e^{-8-7k\eps} = e^{-\Omega(U)}$ and $\delta_k$ we showed is $e^{-\Omega(U)}$ so $3(2e^{-8-7k\eps}+\delta_k) = e^{-\Omega(U)} = n^{-\Omega(1)}.$
\end{proof}

\end{document}